\documentclass[journal]{IEEEtran}
\usepackage{amsfonts}
\usepackage{amsmath}
\usepackage{amsthm}
\usepackage{amssymb}
\usepackage{graphicx}
\usepackage[T1]{fontenc}
\usepackage{supertabular}
\usepackage{longtable}
\usepackage[usenames,dvipsnames]{color}
\usepackage{bbm}
\usepackage{fancyhdr}
\usepackage{breqn}
\usepackage{fixltx2e}
\usepackage{capt-of}
\usepackage{mdframed}
\usepackage{tikz}
\usetikzlibrary{matrix}
\usepackage{endnotes}
\usepackage{soul}

\newtheorem{remark}{Remark}

\newtheorem{proposition}{Proposition}

\newcommand{\mathsym}[1]{}
\newcommand{\unicode}[1]{}

\begin{document}

\title{\color{BlueViolet} A Short Note on P-Value Hacking}
\author{\color{BlueViolet} Nassim Nicholas Taleb\\
Tandon School of Engineering\\
\thanks{Second version, January 2018, First version was March 2015.}
}
\maketitle
\thispagestyle{fancy}
\flushbottom 

\begin{abstract}
We present the expected values from p-value hacking as a choice of the minimum p-value among $m$ independents tests, which can be considerably lower than the "true" p-value, even with a single trial, owing to the extreme skewness of the meta-distribution.
 
We first present an exact probability distribution (meta-distribution) for p-values across ensembles of statistically identical phenomena. We derive the distribution for small samples $2<n \leq n^*\approx 30$ as well as the limiting one as the sample size $n$ becomes large. We also look at the properties of the "power" of a test through the distribution of its inverse for a given p-value and parametrization.

The formulas allow the investigation of the stability of the reproduction of results and "p-hacking" and other aspects of meta-analysis. 

P-values are shown to be extremely skewed and volatile, regardless of the sample size $n$, and vary greatly across repetitions of exactly same protocols under identical stochastic copies of the phenomenon; such volatility makes the minimum $p$ value diverge significantly from the "true" one. Setting the power is shown to offer little remedy unless sample size is increased markedly or the p-value is lowered by at least one order of magnitude.

\end{abstract}


\IEEEPARstart{P}-VALUE hacking, just like an option or other members in the class of convex payoffs, is a function that benefits from the underlying variance and higher moment variability. The researcher or group of researchers have an implicit "option" to pick the most favorable result in $m$ trials, without disclosing the number of attempts, so we tend to get a rosier picture of the end result than reality. The distribution of the minimum p-value and the "optionality" can be made explicit, expressed in a parsimonious formula allowing for the understanding of biases in scientific studies, particularly under environments with high publication pressure.

Assume that we know the "true" p-value, 
 $p_s$, what would its realizations look like across various attempts on statistically identical copies of the phenomena? By true value $p_s$, we mean its expected value by the law of large numbers across an $m$ ensemble of possible samples for the phenomenon under scrutiny, that is $ \frac{1}{m} \sum_{\leq m} p_i \xrightarrow{P} p_s$ (where $\xrightarrow{P} $ denotes convergence in probability). A similar convergence argument can be also made for the corresponding "true median" $p_M$. The distribution of $n$ small samples can be made explicit (albeit with special inverse functions), as well as its parsimonious limiting one for $n$ large, with no other parameter than the median value $p_M$. We were unable to get an explicit form for $p_s$ but we go around it with the use of the median. 
\begin{figure}[h!]
\includegraphics[width=\linewidth]{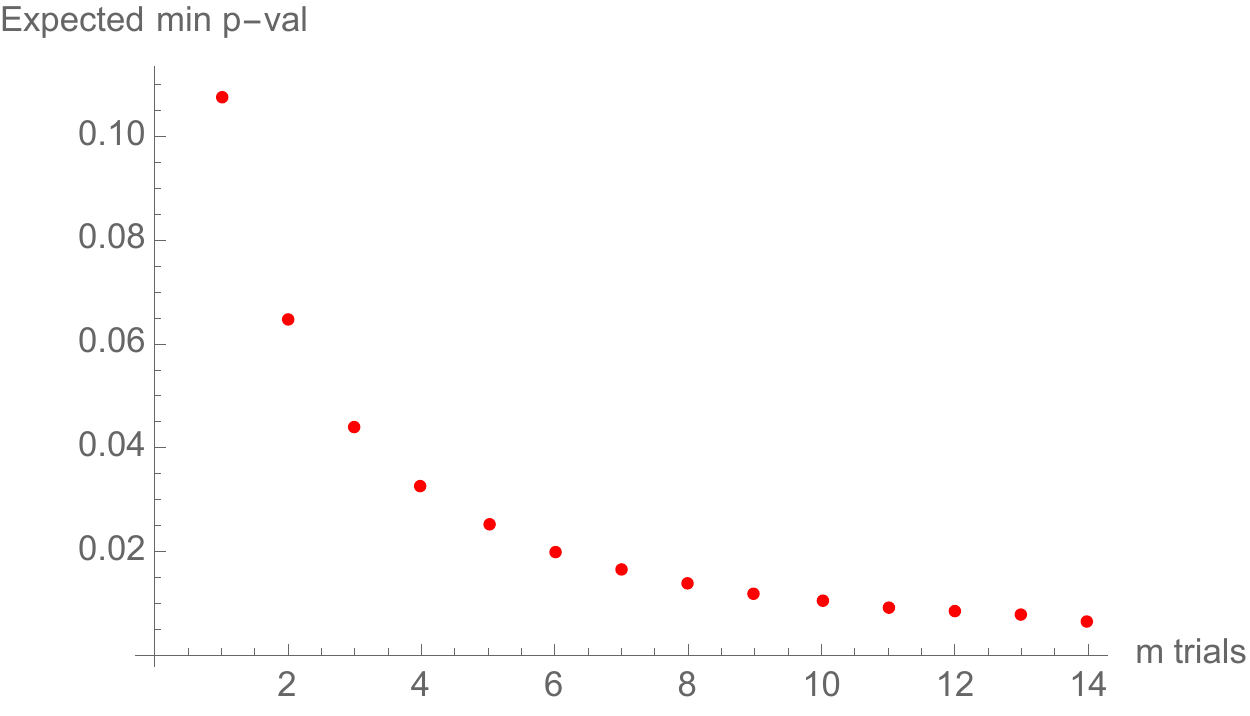}
\caption{The "p-hacking" value across $m$ trials for the "true" median p-value $p_M=.15$ and expected "true" value $p_s=.22$. We can observe how easily one can reach spurious values $<.02$ with a small number of trials.}\label{hacking}
\end{figure}

 \begin{figure}[h!]
\includegraphics[width=\linewidth]{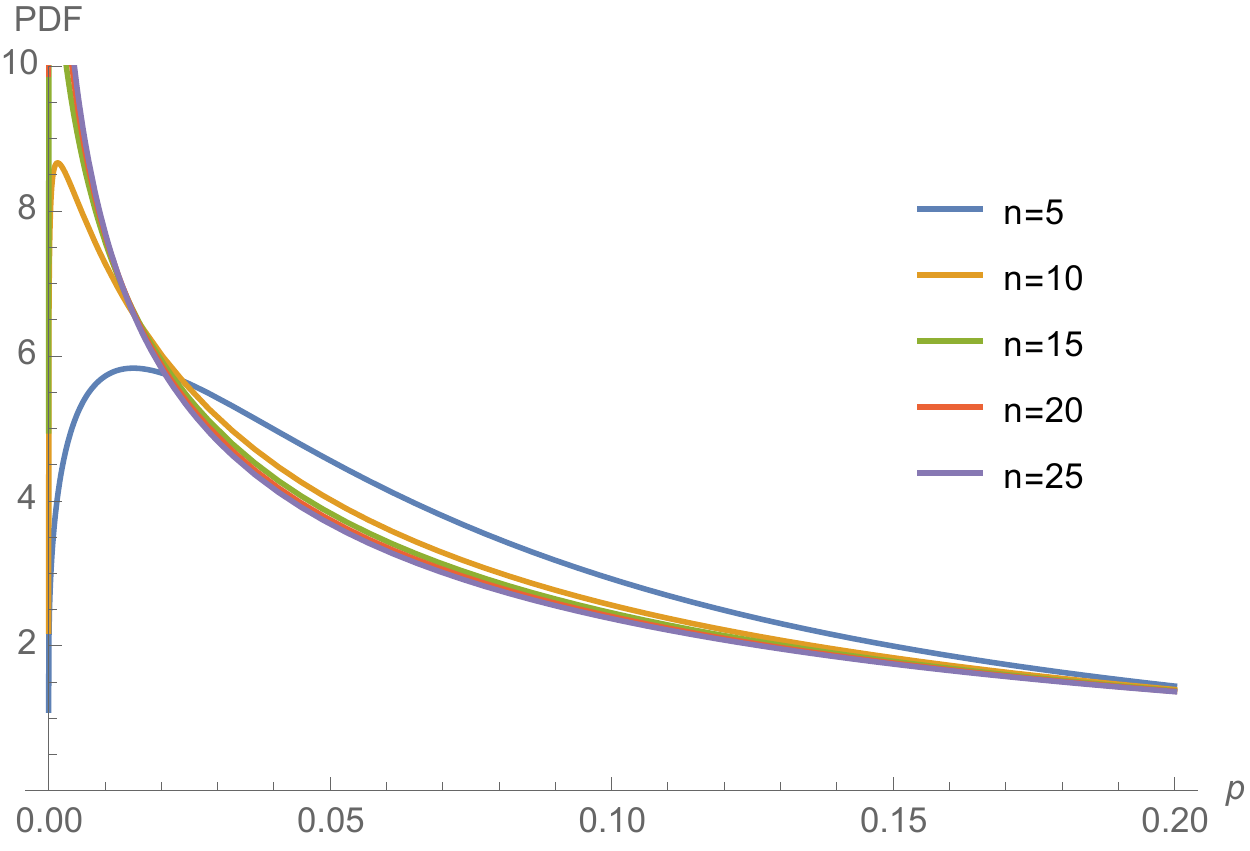}
\caption{The different values for Equ. \ref{firsteq} showing convergence to the limiting distribution. }\label{pval2}
\end{figure}

 It turns out, as we can see in Fig. \ref{pval1} the distribution is extremely asymmetric (right-skewed), to the point where 75\% of the realizations of a "true" p-value of .05 will be <.05 (a borderline situation is $3 \times$ as likely to pass than fail a given protocol), and, what is worse, 60\% of the true p-value of .12 will be below .05. This implies serious gaming and "p-hacking" by researchers, even under a moderate amount of repetition of experiments. 
 \begin{mdframed}[backgroundcolor=yellow!40]
 Although with compact support, the distribution exhibits the attributes of extreme fat-tailedness. For an observed p-value of, say, .02, the "true" p-value is likely to be >.1 (and very possibly close to .2), with a standard deviation >.2 (sic) and a mean deviation of around .35 (sic, sic). Because of the excessive skewness, measures of dispersion in $L^1$ and $L^2$  (and higher norms) vary hardly with $p_s$, so the standard deviation is not proportional, meaning an in-sample $.01$ p-value has a significant probability of having a true value $>.3$. 
  
  \textbf{So clearly we don't know what we are talking about when we talk about p-values.}
 \end{mdframed}

  Earlier attempts for an explicit meta-distribution in the literature were found in \cite{hung1997behavior} and \cite{sackrowitz1999p}, though for situations of Gaussian subordination and less parsimonious parametrization. 
The severity of the problem of \textit{significance of the so-called "statistically significant"} has been discussed in 
\cite{gelman2006difference} and offered a remedy via Bayesian methods in \cite{johnson2013revised}, which in fact recommends the same tightening of standards to p-values $\approx .01$. But the gravity of the extreme skewness of the distribution of p-values is only apparent when one looks at the meta-distribution.

 For notation, we use $n$ for the sample size of a given study and $m$ the number of trials leading to a p-value.
\section{derivation of the metadistribution of p-values}

\begin{proposition}\label{proposition1}
Let $P$ be a random variable $\in [0,1])$ corresponding to the sample-derived one-tailed p-value from the paired T-test statistic (unknown variance) with median value $\mathbb{M}(P)=p_M \in [0,1]$ derived from a sample of $n$  size. The distribution  across the ensemble of statistically identical copies of the sample has for PDF  
$$
\varphi(p;p_M)= \begin{cases}
 	\varphi(p;p_M)_L & \text{for  } p<\frac{1}{2}\\
 	\varphi(p;p_M)_H & \text{for  } p>\frac{1}{2}\\ 	
 \end{cases}
$$

\begin{multline*}
\varphi(p;p_M)_L=
\lambda _p^{\frac{1}{2} (-n-1)}\\
 \sqrt{-\frac{\lambda _p \left(\lambda _{p_M}-1\right)}{\left(\lambda _p-1\right) \lambda _{p_M}-2
   \sqrt{\left(1-\lambda _p\right) \lambda _p} \sqrt{\left(1-\lambda _{p_M}\right) \lambda _{p_M}}+1}} \\
   \left(\frac{1}{\frac{1}{\lambda
   _p}-\frac{2 \sqrt{1-\lambda _p} \sqrt{\lambda _{p_M}}}{\sqrt{\lambda _p} \sqrt{1-\lambda _{p_M}}}+\frac{1}{1-\lambda
   _{p_M}}-1}\right)^{n/2} 
   \end{multline*}
      \begin{multline}
      \varphi(p;p_M)_H=
   \left(1-\lambda '_p\right){}^{\frac{1}{2} (-n-1)} \\
   \left(\frac{\left(\lambda '_p-1\right) \left(\lambda _{p_M}-1\right)}{\lambda '_p
   \left(-\lambda _{p_M}\right)+2 \sqrt{\left(1-\lambda '_p\right) \lambda '_p} \sqrt{\left(1-\lambda _{p_M}\right) \lambda
   _{p_M}}+1}\right){}^{\frac{n+1}{2}} \label{firsteq}
   \end{multline}
\normalfont
where 
$\lambda _p=I_{2 p}^{-1}\left(\frac{n}{2},\frac{1}{2}\right)$,
$\lambda_{p_M}=I_{1-2 p_M}^{-1}\left(\frac{1}{2},\frac{n}{2}\right)$, $\lambda '_p=I_{2 p-1}^{-1}\left(\frac{1}{2},\frac{n}{2}\right)$, and $I^{-1}_{(.)}(.,.)$ is the inverse beta regularized function.
\end{proposition}
\begin{remark}
For p=$\frac{1}{2}$ the distribution doesn't exist in theory, but does in practice and we can work around it with the sequence $p_{m_k}=\frac{1}{2} \pm\frac{1}{k}$, as in the graph showing a convergence to the Uniform distribution on $[0,1]$ in Figure \ref{pval2}. Also note that what is called the "null" hypothesis is effectively a set of measure 0. 
\end{remark}
\begin{proof}
Let $Z$ be a random normalized variable with realizations $\zeta$, from a vector  $\vec{v}$ of $n$ realizations, with sample mean $m_v$, and sample standard deviation $s_v$, $\zeta= \frac{m_v-m_h}{\frac{s_v} {\sqrt{n}}}$ (where $m_h$ is the level it is tested against), hence assumed to $\thicksim $ Student T with $n$ degrees of freedom, and, crucially, supposed to deliver a mean of $\bar{\zeta}$, $$f(\zeta;\bar{\zeta})=\frac{\left(\frac{n}{(\bar{\zeta}-\zeta )^2+n}\right)^{\frac{n+1}{2}}}{\sqrt{n} B\left(\frac{n}{2},\frac{1}{2}\right)}$$ where B(.,.) is the standard beta function. Let $g(.)$ be the one-tailed survival function of the Student T distribution with zero mean and $n$ degrees of freedom: 
$$g(\zeta)=  \mathbb{P}(Z>\zeta)=
\begin{cases}
 \frac{1}{2} I_{\frac{n}{\zeta^2+n}}\left(\frac{n}{2},\frac{1}{2}\right) & \zeta\geq 0 \\
 \frac{1}{2} \left(I_{\frac{\zeta^2}{\zeta^2+n}}\left(\frac{1}{2},\frac{n}{2}\right)+1\right) & \zeta <0
 \end{cases}$$
where $I_{(.,.)}$ is the incomplete Beta function.

We now look for the distribution of $g \circ f(\zeta)$. 
Given that g(.) is a legit Borel function, and naming $p$ the probability as a random variable, we have by a standard result for the transformation:
 
$$\varphi(p,\bar{\zeta})=\frac{f\left(g^{(-1)}(p)\right)}{|g'\left(g^{(-1)}(p)\right)|}$$

We can convert $\bar{\zeta}$ into the corresponding median survival probability because of symmetry of $Z$. Since one half the observations fall on either side of $\bar{\zeta}$, we can ascertain that the transformation is median preserving: $g(\bar{\zeta})=\frac{1}{2}$, hence $\varphi(p_M,.) =\frac{1}{2}$. Hence we end up having  $\{\bar{\zeta}:\frac{1}{2} I_{\frac{n}{\bar{\zeta}^2+n}}\left(\frac{n}{2},\frac{1}{2}\right)=p_M \}$ (positive case) and $\{\bar{\zeta}:\frac{1}{2} \left(I_{\frac{\zeta^2}{\zeta^2+n}}\left(\frac{1}{2},\frac{n}{2}\right)+1\right) =p_M\}$ (negative case). Replacing we get Eq.\ref{firsteq} and Proposition \ref{proposition1} is done.


\end{proof}

We note that $n$ does not increase significance, since p-values are computed from normalized variables (hence the universality of the meta-distribution); a high $n$ corresponds to an increased convergence to the Gaussian. For large $n$, we can prove the following proposition:

\begin{proposition} Under the same assumptions as above, the limiting distribution for $\varphi(.)$:
\begin{equation}
\lim_{n\to \infty} \varphi(p;p_M)=e^{ -\text{erfc}^{-1}(2 p_M) \left(\text{erfc}^{-1}(2 p_M)-2 \text{erfc}^{-1}(2 p)\right)}\label{equationgauss}
	\end{equation}
where erfc(.) is the complementary error function and $erfc(.)^{-1}$ its inverse.\label{proposition2}

The limiting CDF $\Phi(.)$ 
\begin{equation}
	\Phi(k;p_M)=\frac{1}{2} \text{erfc}\left(\text{erf}^{-1}(1-2 k)-\text{erf}^{-1}(1-2 p_M)\right)
\end{equation}

\end{proposition}
\begin{proof}
For large $n$, the distribution of $Z= \frac{m_v}{\frac{s_v} {\sqrt{n}}}$ becomes that of a Gaussian, and the one-tailed survival function 
$g(.)=\frac{1}{2}\text{erfc}\left(\frac{\zeta }{\sqrt{2}}\right)$,
$\zeta(p) \to \sqrt{2} \text{erfc}^{-1}(p)$.	
\end{proof}

 \begin{figure}[h!]
\includegraphics[width=\linewidth]{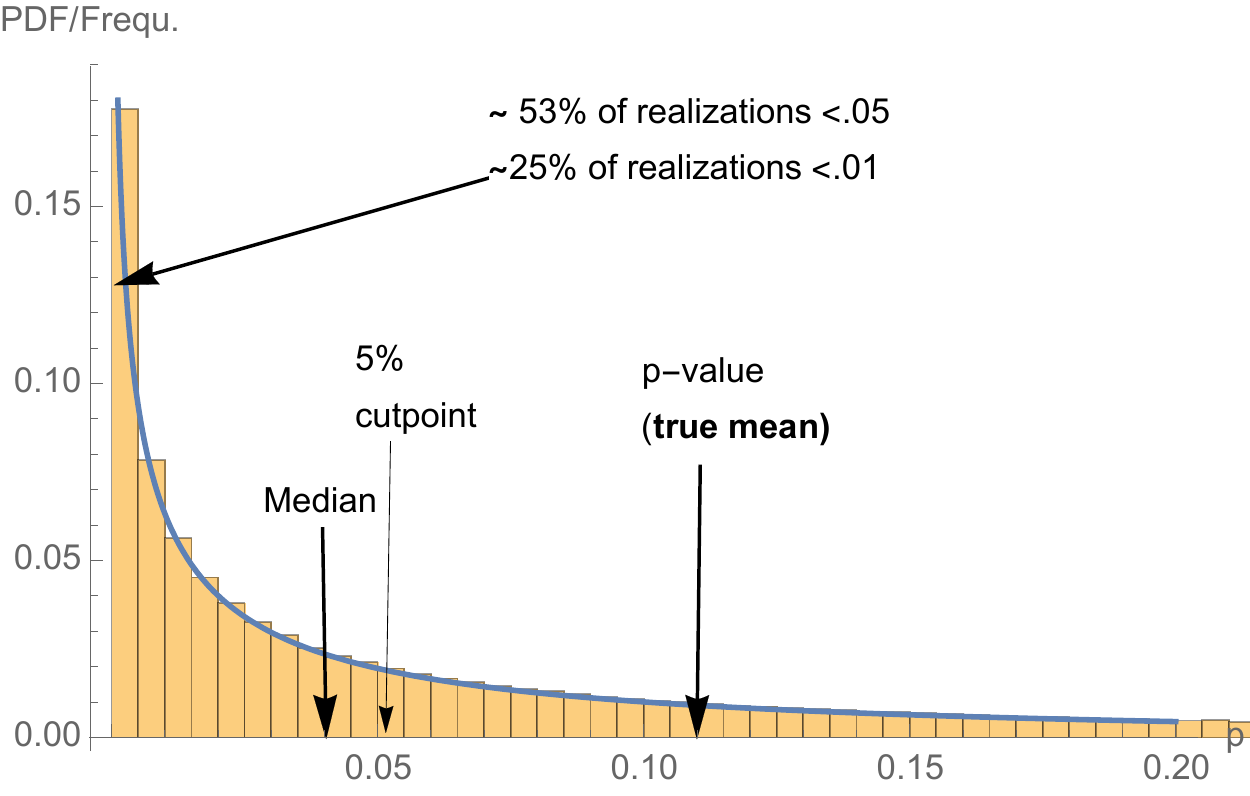}
\caption{The probability distribution of a one-tailed p-value with expected value .11 generated by Monte Carlo (histogram) as well as analytically with $\varphi(.)$ (the solid line). We draw all possible subsamples from an ensemble with given properties. The excessive skewness of the distribution makes the average value considerably higher than most observations, hence causing illusions of "statistical significance".}\label{pval1}
\end{figure}

\begin{figure}[h!]
\includegraphics[width=\linewidth]{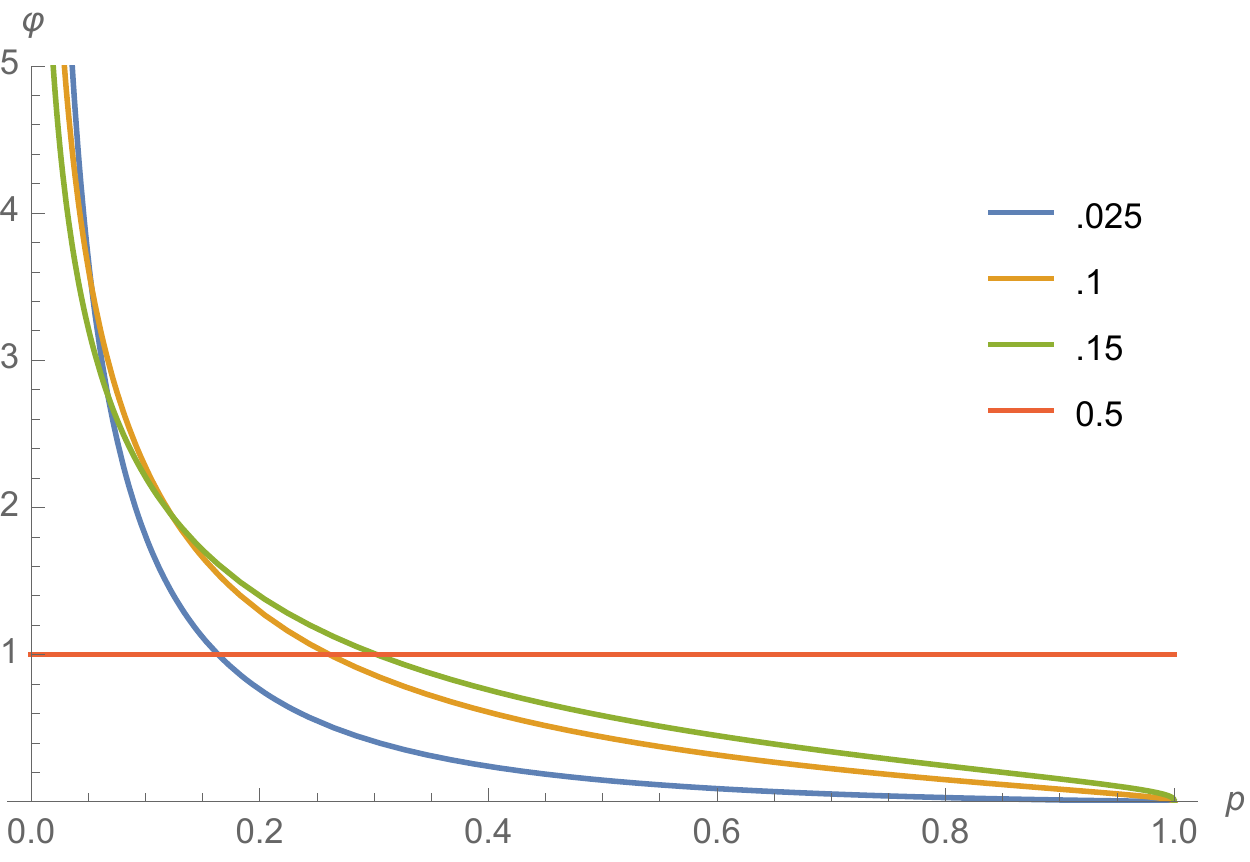}
\caption{The probability distribution of p at different values of $p_M$. We observe how $p_M = \frac{1}{2}$ leads to a uniform distribution.}\label{pval2}
\end{figure}

This limiting distribution applies for paired tests with known or assumed sample variance since the  test becomes a Gaussian variable, equivalent to the convergence of the T-test (Student T) to the Gaussian when $n$ is large.
\begin{remark}
For values of $p$ close to 0, $\varphi$ in Equ. \ref{equationgauss} can be usefully calculated as:
\begin{multline}
\varphi(p;p_M) =\sqrt{2 \pi } p_M \sqrt{\log \left(\frac{1}{2 \pi  p_M^2}\right)}\\
 e^{\sqrt{-\log \left(2 \pi  \log \left(\frac{1}{2 \pi  p^2}\right)\right)-2 \log (p)} \sqrt{-\log \left(2 \pi  \log \left(\frac{1}{2 \pi  p_M^2}\right)\right)-2 \log \left(p_M\right)}}\\
 +O(p^2).
\end{multline}
The approximation works more precisely for the band of relevant values $0<p<\frac{1}{2\pi}$.\end{remark}

From this we can get numerical results for convolutions of $\varphi$ using the Fourier Transform or similar methods.
\section{P-Value Hacking}
We can and get the distribution of the minimum p-value per $m$ trials across statistically identical situations thus get an idea of "p-hacking", defined as attempts by researchers to get the lowest p-values of many experiments, or try until one of the tests produces statistical significance. 
\begin{proposition}

The distribution of the minimum of $m$ observations of statistically identical p-values becomes (under the limiting distribution of proposition \ref{proposition2}):
\begin{multline}
\varphi_m(p;p_M)=m \,e^{\text{erfc}^{-1}(2 p_M) \left(2 \text{erfc}^{-1}(2 p)-\text{erfc}^{-1}(2 p_M)\right)}\\
 \left(1-\frac{1}{2} \text{erfc}\left(\text{erfc}^{-1}(2 p)-\text{erfc}^{-1}(2 p_M)\right)\right)^{m-1}
\end{multline}	
\end{proposition}

\begin{proof}

$P\left(p_1>p , p_2>p, \ldots,p_m>p\right)= \bigcap_{i=1}^n  \Phi(p_i)= \bar{\Phi}(p)^m$. Taking the first derivative  we get the result. 
\end{proof}
Outside the limiting distribution: we integrate numerically for different values of m as shown in figure \ref{hacking}. So, more precisely, for $m$ trials, the expectation is calculated as: 
	 $$\mathbb{E}(p_{min})= \int_0^1-m \; \varphi (p;p_M) \left(\int_0^p \varphi (u,.) \, \,\mathrm{d}u\right)^{m-1} \mathrm{d}p$$

\section{Other Derivations}
\subsection*{Inverse Power of Test}
Let $\beta$ be the power of a test for a given p-value $p$, for random draws X from unobserved parameter $\theta$ and a sample size of $n$. To gauge the reliability of $\beta$ as a true measure of power, we perform an inverse problem: 

\begin{tikzpicture}[every node/.style={midway}]
\matrix[column sep={4em,between origins},
        row sep={2em}] at (0,0)
{ \node(R)   {$\beta$}  ; & \node(S) {$X_{\theta,p,n}$}; \\
  \node(R/I) {$\beta^{-1}(X)$};                   \\};
\draw[<->] (R/I) -- (R) node[anchor=east]  {$\Delta$};
\draw[<-] (R/I) -- (S) node[anchor=west]  {$$};
\draw[->] (R)   -- (S) node[anchor=east] {$$};
\end{tikzpicture}

\begin{proposition}
	Let $\beta_c$ be the projection of the power of the test from the realizations assumed to be student T distributed and evaluated under the parameter $\theta$. We have
$$
\Phi  (\beta _c)= \begin{cases}
 	\Phi  (\beta _c)_L & \text{for  } \beta _c<\frac{1}{2}\\
 	\Phi  (\beta _c)_H & \text{for  } \beta _c>\frac{1}{2}\\ 	
 \end{cases}
$$
where
\begin{multline}
\Phi  (\beta _c)_L=\sqrt{1-\gamma _1} \gamma _1^{-\frac{n}{2}} \\
\frac{\left(-\frac{\gamma _1}{2 \sqrt{\frac{1}{\gamma _3}-1} \sqrt{-\left(\gamma _1-1\right) \gamma _1}-2 \sqrt{-\left(\gamma _1-1\right) \gamma _1}+\gamma _1 \left(2 \sqrt{\frac{1}{\gamma _3}-1}-\frac{1}{\gamma _3}\right)-1}\right){}^{\frac{n+1}{2}}}{\sqrt{-\left(\gamma _1-1\right) \gamma _1}}	
\end{multline}
\begin{multline}
\Phi  (\beta _c)_H=	 \sqrt{\gamma _2} \left(1-\gamma _2\right)^{-\frac{n}{2}} B\left(\frac{1}{2},\frac{n}{2}\right)\\
	 \frac{ \left(\frac{1}{\frac{-2 \left(\sqrt{-\left(\gamma _2-1\right) \gamma _2}+\gamma _2\right) \sqrt{\frac{1}{\gamma _3}-1}+2 \sqrt{\frac{1}{\gamma _3}-1}+2 \sqrt{-\left(\gamma _2-1\right) \gamma _2}-1}{\gamma _2-1}+\frac{1}{\gamma _3}}\right){}^{\frac{n+1}{2}}}{\sqrt{-\left(\gamma _2-1\right) \gamma _2} B\left(\frac{n}{2},\frac{1}{2}\right)} 
\end{multline}
where
$\gamma _1= I_{2 \beta _c}^{-1}\left(\frac{n}{2},\frac{1}{2}\right)$,
$\gamma _2=I_{2 \beta _c-1}^{-1}\left(\frac{1}{2},\frac{n}{2}\right)$, and
$\gamma _3=I_{\left(1,2 p_s-1\right)}^{-1}\left(\frac{n}{2},\frac{1}{2}\right)$.
\end{proposition}

\section{Application and Conclusion}

\begin{itemize}
\item One can safely see that under such stochasticity for the realizations of p-values and the distribution of its minimum, to get what a scientist would expect from a 5\% confidence level (and the inferences they get from it), one needs a p-value of at least one order of magnitude smaller. 
\item Attempts at replicating papers, such as the open science project \cite{open2015estimating}, should consider a margin of error in \textit{its own} procedure and a pronounced bias towards favorable results (Type-I error). There should be no surprise that a previously deemed significant test fails during replication --in fact it is the replication of results deemed significant at a close margin that should be surprising. 
\item The "power" of a test has the same problem unless one either lowers p-values or sets the test at higher levels, such at .99.
\end{itemize}

\section*{Acknowledgment}
Marco Avellaneda, Pasquale Cirillo, Yaneer Bar-Yam, friendly people on twitter, less friendly verbagiastic psychologists on twitter, ...
\bibliographystyle{IEEEtran}
\bibliography{/Users/nntaleb/Dropbox/Central-bibliography}

\begin{thebibliography}{1}
\providecommand{\url}[1]{#1}
\csname url@samestyle\endcsname
\providecommand{\newblock}{\relax}
\providecommand{\bibinfo}[2]{#2}
\providecommand{\BIBentrySTDinterwordspacing}{\spaceskip=0pt\relax}
\providecommand{\BIBentryALTinterwordstretchfactor}{4}
\providecommand{\BIBentryALTinterwordspacing}{\spaceskip=\fontdimen2\font plus
\BIBentryALTinterwordstretchfactor\fontdimen3\font minus
  \fontdimen4\font\relax}
\providecommand{\BIBforeignlanguage}[2]{{%
\expandafter\ifx\csname l@#1\endcsname\relax
\typeout{** WARNING: IEEEtran.bst: No hyphenation pattern has been}%
\typeout{** loaded for the language `#1'. Using the pattern for}%
\typeout{** the default language instead.}%
\else
\language=\csname l@#1\endcsname
\fi
#2}}
\providecommand{\BIBdecl}{\relax}
\BIBdecl

\bibitem{hung1997behavior}
H.~J. Hung, R.~T. O'Neill, P.~Bauer, and K.~Kohne, ``The behavior of the
  p-value when the alternative hypothesis is true,'' \emph{Biometrics}, pp.
  11--22, 1997.

\bibitem{sackrowitz1999p}
H.~Sackrowitz and E.~Samuel-Cahn, ``P values as random variables---expected p
  values,'' \emph{The American Statistician}, vol.~53, no.~4, pp. 326--331,
  1999.

\bibitem{gelman2006difference}
A.~Gelman and H.~Stern, ``The difference between ``significant'' and ``not
  significant'' is not itself statistically significant,'' \emph{The American
  Statistician}, vol.~60, no.~4, pp. 328--331, 2006.

\bibitem{johnson2013revised}
V.~E. Johnson, ``Revised standards for statistical evidence,''
  \emph{Proceedings of the National Academy of Sciences}, vol. 110, no.~48, pp.
  19\,313--19\,317, 2013.

\bibitem{open2015estimating}
O.~S. Collaboration \emph{et~al.}, ``Estimating the reproducibility of
  psychological science,'' \emph{Science}, vol. 349, no. 6251, p. aac4716,
  2015.

\end{thebibliography}

\end{document}